\documentclass[10pt,showpacs,floatfix,letterpaper,aps,prl,reprint, showkeys]{revtex4-1}
\pdfoutput=1
\usepackage{florjanc}

\begin{document}
\title{Continuous decomposition of quantum measurements via Hamiltonian feedback}
\author{Jan Florjanczyk} \email{florjanc@usc.edu}
\author{Todd A. Brun} \email{tbrun@usc.edu}
\affiliation{Center for Quantum Information Science and Technology, \\
Communication Sciences Institute, Department of Electrical Engineering, \\
University of Southern California Los Angeles, CA 90089, USA.}
\date{\today}

\pacs{03.65.Aa, 03.65.Ta}

\keywords{quantum continuous measurement, quantum feedback control, Jordan algebras}
\begin{abstract}
We characterize the set of generalized quantum measurements that can be decomposed into a continuous measurement process using a stream of probe qubits and a tunable interaction Hamiltonian. Each probe in the stream interacts weakly with the target quantum system, then is measured projectively in a standard basis. This measurement result is used in a closed feedback loop to tune the interaction Hamiltonian for the next probe. The resulting evolution is a stochastic process with the structure of a one-dimensional random walk. To maintain this structure, and require that at long times the measurement outcomes be independent of the path, the allowed interaction Hamiltonians must lie in a restricted set, such that the Hamiltonian terms on the target system form a finite dimensional Jordan algebra. This algebraic structure of the interaction Hamiltonians yields a large class of generalized measurements that can be continuously performed by our scheme, and we fully describe this set.
\end{abstract}
\maketitle

Many quantum systems either exhibit slow measurement read-out times or can only be probed weakly. Under such conditions, it is natural to monitor the systems continuously while simultaneously exerting some closed-loop feedback. Experiments can already be performed with such low latency that feedback can be performed continuously in real time~\cite{rydberg, rydberg-feedback}. While generalized continuous measurements have been studied~\cite{weakuniversal, generalizedstochastic}, in most systems the diffusive weak measurements~\cite{simplemodel} that constitute the continuous process must be applied via coupling to a probe system. Previously, we've studied a system with closed-loop feedback applied to a stream of probe qubits interacting with the system by a fixed Hamiltonian~\cite{constham}. Here, we investigate the possibilities that arise from closed-loop feedback when the interaction Hamiltonian is itself subject to control.

A key feature of~\cite{constham} was the derivation of a reversibility equation which was used to restrict the class of measurements that admitted a continuous decomposition. This equation is necessary again in this work to ensure that the final measurement at long times is independent of the details of the path. Although we'll restrict our analysis to qubit probes and two-outcome measurements, we note that general two-outcome measurements are sufficient building blocks for $n$-outcome measurements~\cite{nonlocality, weakuniversal}.

\begin{figure}
  \includegraphics[width=0.45\textwidth]{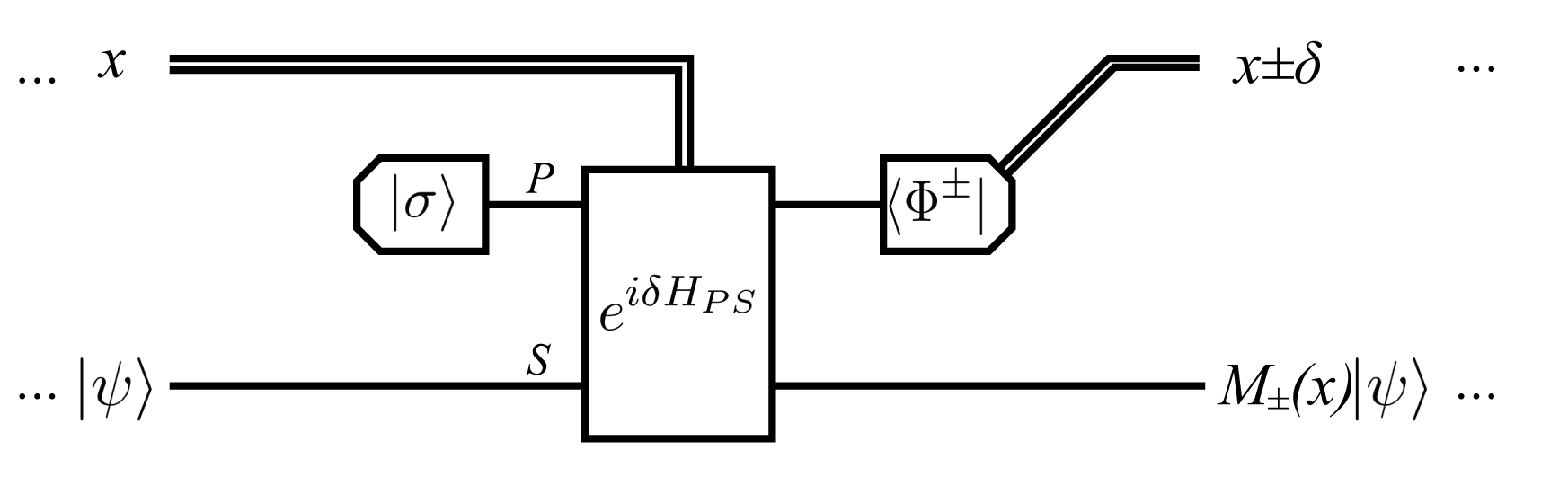}
  \caption{The system $S$ is continuously measured. At each timestep, we perform a weak measurement by preparing the probe $\ket{\s}$ and tuning the interaction Hamiltonian $H_{PS}(x)$ based on a pointer variable $x$. The system and probe interact for a short time $\d$ and the probe is measured in an orthogonal detector basis $\bra{\Phi^{\pm}}$. The measurement result from the detector is used to update the pointer variable from $x$ to $x \pm \d$ and the procedure is repeated with the new value.}
    \label{fig:circuit}
\end{figure}

Consider a quantum system $S$ undergoing a stochastic evolution driven by two-outcome diffusive weak measurements. The outcome of any particular step during the evolution is one of two \emph{weak measurement step operators} $M_{\pm}(x)$. These step operators are functions of a pointer variable $x$ which updates with each outcome. The exact feedback scheme is illustrated in Figure~\ref{fig:circuit}, and the process terminates when $x$ reaches a fixed constant $\pm X$. The reversibility condition can be written
\begin{equation}
    \label{eqn:Revers}
    M_{\mp}(x \pm \d) M_{\pm} (x) \propto I .
\end{equation}
From the above equation, two consecutive outcomes that step ``forward'' from $x$ to $x+\d$ then ``backward'' from $x+\d$ to $x$ have no net effect on $\ket{\psi}$. However, rather than track the evolution of $\ket{\psi}$ directly, we can express the total action of our procedure as the \emph{total walk operator}
\begin{equation}
  \label{eqn:totalwalkoperator}
  M(x) = \lim_{\d \rightarrow 0} \prod_{j=1}^{ \lfloor x / \d \rfloor} M_+(j\d),
\end{equation}
For negative values of $x$ we replace $M_+$ with $M_-$ above. We identify the \emph{endpoint} operators $M_1 \propto M(X)$ and $M_2 \propto M(-X)$ with the final measurement being decomposed by this process. We will consider a simple model of probe-state interaction that will generate $M_{\pm}(x)$. In~\cite{constham} we found that weak measurements with qubit probes had to form a \emph{probe-basis} on the qubit Hilbert space. In particular, we required that the probe state, the orthogonal quantum states of the detector, and the probe eigenstates of the interaction Hamiltonian, have mutually orthonormal representations on the Bloch sphere. For this reason, we choose the interaction Hamiltonian to be $H_{PS} = Y_P \ox \hat{\e}(x)$, the probe state to be $\ket{0}$, and the detector states to be $\bra{\pm}$. The operator acting on the system $S$ is $\hat{\e}$ and is defined to be an $x$-dependent linear combination of $d$ constant Hamiltonian terms,
\begin{equation}
  \hat{\e}(x) = \sum_{i=0}^d p_i(x) H_i . \label{eqn:epsilon}
\end{equation}
The weak measurement step operators of Figure~\ref{fig:circuit} then become
\begin{align} 
  M_{\pm}(x) & = \bra{\pm} e^{i \d H_{PS}(x)} \ket{0} \\
  & \approx \frac{1}{\sqrt{2}} I \mp \frac{\d}{\sqrt{2}} \hat{\e}(x) - \frac{\d^2}{2\sqrt{2}} \hat{\e}^2(x). \nonumber
\end{align}
The reversibility condition of Eq.~\ref{eqn:Revers} can now be rewritten in terms of $\hat{\e}(x)$. Note that this condition need only be satisfied up to $O(\d^2)$ since the random walk induced on the pointer variable $x$ will take $O(N^2)$ steps to converge when $N = \lfloor X/\d \rfloor$. Collecting terms by orders of $\d$ yields
\begin{equation*}
  M_{\mp}(x \pm \d) M_{\pm}(x) = \frac{I}{2} + \frac{\d^2}{2} \lp \partial_x \hat{\e}(x) - 2 \hat{\e}^2(x)\rp + O(\d^3).
\end{equation*}
Let $\a(x)$ be the proportionality constant in Eq.~\eqref{eqn:Revers}. We find that the reversibility equation reduces to
\begin{equation}
  \label{eqn:revers}
  \partial_x \hat{\e}(x) = 2 \hat{\e}^2(x) + \a(x) I.
\end{equation}
In the derivations that follow, we will ignore the $\a(x) I$ term, as it will not change the class of measurements that satisfy the reversibility equation. (In practice, the term can be reintroduced to help find bounded solutions.)

Consider the set of controls that appear in Eq.~\eqref{eqn:epsilon}. Without loss of generality, we can always assume that $H_0 = I$ since the action of $I$ is equivalent to an overall phase on the probe system. The reversibility equation in Eq.~\eqref{eqn:revers} can then be rewritten as
\begin{equation}
  \sum_{k=0}^d \partial_x p_k(x) H_k = \sum_{i, j=0}^d p_i(x) p_j(x) \frac{1}{2} \lb H_i,  H_j \rb \label{eqn:symmetrizing}.
\end{equation}
where $\lb \cdot, \cdot \rb$ is the anti-commutator. It will be useful to introduce the tensor $\G_{ij}^k$ for expressing the action of the anti-commutator on the matrices $H_i$. In particular,
\begin{equation}
  \frac{1}{2} \lb H_i, H_j \rb = \sum_{k=0}^{n(n-1)/2} \G_{ij}^k H_k .
\end{equation}
We choose the matrices $H_i$ for $i>d$ such that they form a basis for $\mcal{H}_n \lp \mbC \rp$, the space of all $n$-dimensional complex Hermitian matrices. We will use $\G^{(k)}$ to denote the matrix resulting from fixing the index $k$. The reversibility equation Eq.~\eqref{eqn:revers} can then be read as
\begin{equation} \label{eqn:di_xp}
  \lb \begin{array}{rlcl}
    \partial_x p_k &= \vec{p}^T \G^{(k)} \vec{p} & \hspace{0.15in} & 0 \leq k \leq d \\
    0 &= \vec{p}^T \G^{(k)} \vec{p} & \hspace{0.15in} & d < k .
  \end{array} \right.
\end{equation}

We now present our main result which characterizes solutions to the above equations. Let us denote $\mbF = \spa{H_i}$ so that $\hat{\e} \in \mbF$. We prove the following lemma about solutions to Eq.~\eqref{eqn:di_xp}.

\begin{lem} \label{lem:closure}
Any solution $\hat{\e}(x)$ to Eq.~\eqref{eqn:di_xp} must lie entirely in $\mbV$, a subspace of $\mbF$ that is closed under anti-commutation.
\end{lem}
\begin{proof}
We note that if $\mbF$ is already closed under anti-commutation, then the reversibility equation reduces to an initial value problem in terms of the control coefficients $\vec{p}(x)$ at $x=0$. However if $\mbF$ is not closed under anti-commutation, then we must characterize the set of vectors $\vec{p}$ such that Eq.~\eqref{eqn:di_xp} is satisfied. To do so, consider choosing any $k>d$ and solving the associated equation $\vec{p}^T \G^{(k)} \vec{p} = 0$. Note that the matrix $\G^{(k)}$ is symmetric and defines a quadratic space over $\mbR^d$. Every quadratic space admits a Witt decomposition~\cite{omeara} which in our case is
\begin{equation}
    \lp \G^{(k)}, \mbR^d \rp \cong \bigoplus_{i=0}^{N} W_i \oplus V_0 \oplus V' .
\end{equation}
In the above, $W_i$ are hyperbolic planes, $V_0$ is the nullspace of $\G^{(k)}$, and $V'$ is an anisotropic subspace of $\mbR^d$. Solutions to $\vec{x}^T W_i \vec{x} = 0$ are $\spa{[1,1]} \cup \spa{[1,-1]}$. Additionally, there are no vectors which satisfy $\vec{x}^T V' \vec{x} = 0$ for the anisotropic subspace $V'$. Let $T^{(k)}$ be the isomorphism of $\lp \G^{(k)}, \mbR^d \rp$ to $\lp I_d, \mbR^d \rp$ and $\vec{p} = T^{(k)} \vec{q}$. Then possible solutions to $\vec{p}^T \G^{(k)} \vec{p} = 0$ must lie in
\begin{equation}
    V = T^{(k)} \lp \bigoplus_{i=0}^{N} \spa{\ls 1, x_i \rs} \oplus V_0 \rp
\end{equation}
for a fixed choice of $x_i = \pm 1$. To fully solve Eq.~\eqref{eqn:di_xp} we must now recurse the above procedure. At each step we restrict $\vec{p}$ to lie in the subspace $V$ defined by a choice of $x_i$. We then define a new matrix basis for the controls restricted to $V$ and generate a new set of $\G^{(k)}$ matrices. We then choose a new $k$ and decompose $V$ using $\G^{(k)}$. Since the order in which the $k$ are chosen will affect the form of $V$, it is also important to enumerate all sequences of choices of $k$ and $x_i$. This procedure terminates when the vector space of Hermitian matrices $\mbV$ formed from $V$ is closed under anti-commutation. Furthermore, since the Witt decomposition is unique (up to isometries of $V'$), we can guarantee that this procedure lists all closed subspaces contained in $\mbF$. It remains only to show that if $\vec{p}(0) \in V$ for a particular sequence of choices of $k$ and $x_i$, that $\vec{p}(x)$ will remain in the same subspace for all other values of $x$. This follows directly, however, from the fact that if $\hat{\e}(x) \in \mbV$ then $\hat{\e}^2(x) \in \mbV$ and so $\di_x \hat{\e} \in \mbV$.
\end{proof}

Lemma~\ref{lem:closure} establishes that in order to solve the reversibility equation, one must use a set of controls whose span is closed under anti-commutation. The proof of the lemma also includes an implicit algorithm for finding closed subspaces given a set of Hermitian matrices. The next lemma gives the structure of the subspaces enumerated by lemma~\ref{lem:closure}.

\begin{lem} \label{lem:wedderburn}
The $\hat{\e}(x)$ operator has the form
\begin{equation}
    \label{eqn:eps_decomp}
    \hat{\e}(x) = \bigoplus_{l=1}^{S(\mbV)} U_l(x) D_l(x) U_l^{\dag}(x) .
\end{equation}
where $S(\mbV)$ is the number of simple components of the algebra $\mbV$ (with anti-commutation as its product), and $D_l(x)$ and $U_l(x)$ correspond to the $l^{\mathrm{th}}$ simple component and are given by Table~\ref{tab:repr}.
\end{lem}
\begin{proof}
We begin by identifying $\mbV$ as a finite-dimensional Jordan algebra. Every such algebra accepts a Wedderburn-type decomposition~\cite{albert-wedderburn, penico-wedderburn},
\begin{equation}
    \label{eqn:wedderburn}
  \mbV \cong \bigoplus_{l=1}^{S(\mbV)} \mbB_l, 
\end{equation}
where $S(\mbV)$ is the number of simple components $\mbB_l$ of $\mbV$. A classification of all finite-dimensional simple Jordan algebras was given by Jordan, von Neumann, and Wigner~\cite{jordanmain}. The three types of Jordan algebras that can be found in our decomposition are the self-adjoint real, complex, and quaternionic matrices. The isomorphism in Eq.~\eqref{eqn:wedderburn} leaves a lot of freedom in terms of how to represent each of these simple components by Hamiltonian terms. We summarize the possible representations in Table~\ref{tab:repr}. (Note that the exceptional Albert algebra is absent, since octonions do not have a matrix representation over $\mbR$ or $\mbC$~\cite{albert}). Since $\mbV$ can be written as a direct sum, we can also write
\[ \hat{\e} = \bigoplus_{l=1}^{S(\mbV)} \hat{\e}_l(x). \]
Each operator in the direct sum can, in turn, be diagonalized to yield the form in the statement of the lemma.
\begin{table}
    \renewcommand\arraystretch{1.5}
    \renewcommand\tabcolsep{6pt}
    \begin{tabular}{|c|c|c|c|c|}
        \hline      Block $\mbB_l$                                      & $D_l(x)$                 & $U_l(x)$ \\
        \hline
        \hline    $\mcal{H}_{n}(\mbR)$                                & $\diag{\mbR^n}$      & $SO(n)$ \\
        \hline    $\mcal{H}_{n}(\mbC)$                                & $\diag{\mbR^n}$          & $SU(n)$ \\
        \hline    $\mcal{H}_{n}(\mbC) \cong \mcal{H}_{2n}(\mbR)$      & $\diag{\mbR^n}\ox I_2$   & $SO(n) \ox SO(2)$ \\
        \hline    $\mcal{H}_{n}(\mbH) \cong \mcal{H}_{2n}(\mbC)$      & $\diag{\mbR^n}\ox I_2$   & $SU(n) \ox SU(2)$ \\
        \hline    $\mcal{H}_{n}(\mbH) \cong \mcal{H}_{4n}(\mbR)$      & $\diag{\mbR^n}\ox I_4$   & $SO(n) \ox SO(4)$ \\
        \hline
    \end{tabular}
    \caption{We list all rank-$n$ representations of Jordan algebras that can be embedded into a span of Hermitian matrices. The third representation corresponds to the $2$-dimensional embedding of $\mbC$ into $\mbR$. The fourth and fifth representations correspond to $2$- and $4$-dimensional embeddings of $\mbH$ into $\mbC$ and $\mbR$. The notation $\diag{\mbR^n}$ refers to the space of $n$-dimensional diagonal matrices.}
    \label{tab:repr}
\end{table}
\end{proof}

It remains to describe the form of the endpoints of the continuous process $M_1$, $M_2$. We use the reversibility and propagation equations to solve for them directly in the following lemma.

\begin{lem} \label{lem:diagonal}
The $\hat{\e}(x)$ operator and the total walk operator $M(x)$ are simultaneously diagonalizable.
\end{lem}
\begin{proof}
We begin by noting that Eq.~\ref{eqn:revers} can be solved for individual blocks $\hat{\e}_l(x)$ which yield,
\begin{equation}
  \di_x \lp U_l(x) D_l(x) U_l^{\dag}(x) \rp  = 2 \lp U_l(x) D_l(x) U_l^{\dag}(x) \rp^2 .
\end{equation}
Since $U_l(x)$ is a unitary matrix we can write it as the exponent of a Hermitian matrix $G_l(x)$ and we note that $U_l^{\dag}(x) \partial_x U_l(x)= i \partial_x G_l(x)$. This reduces the above equation to
\begin{equation}
    \label{eqn:diag}
    i \ls \partial_x G_l, D_l \rs + \partial_x D_l = 2 D_l^2 .
\end{equation}
The entries of the commutator term are
\begin{equation}
    \lp \ls \partial_x G_l, D_l \rs \rp_{ij} = \di_x g^{(l)}_{ij} \lp d^{(l)}_i - d_j^{(l)} \rp
\end{equation}
from which we can infer that the diagonal entries of the commutator term are $0$ if $d^{(l)}_i \neq d^{(l)}_j$. We can show that $d^{(l)}_i(x)$ has the solution $\tanh(x-c^{(l)}_i)$, and thus for any $i,j$ such that $c^{(l)}_i \neq c^{(l)}_j$, $g_{ij}^{(l)}$ is constant. All together, Eq.~\eqref{eqn:diag} has the solution
\begin{align}
  \label{eqn:d_l}
  \lb \begin{array}{lcl}
  d^{(l)}_i(x) = \tanh(x-c^{(l)}_i) & \hspace{0.1in} & \forall \; i , \\
  g^{(l)}_{ij}(x) = g^{(l)}_{ij}(0) & & \forall\; i,j \; : \; c^{(l)}_i \neq c^{l)}_j , \\
  g^{(l)}_{ij}(x) = g^{(l)}_{ij}(x) & & \forall\; i,j \; : \; c^{(l)}_i = c^{l)}_j .
  \end{array}\right.
\end{align}
The solution above for $d_i^{(l)}(x)$ is found by reintroducing the $\a(x) I$ term to our equations (which we've ignored thus far). We note that in the cases where $c^{(l)}_i = c^{(l)}_j$, $g_{ij}^{(l)}$ need not be constant. This, however, does not affect the form of $\hat{\e}$, or of $M_1$, $M_2$.

We now turn our attention to the total walk operator given in Eq.~\eqref{eqn:totalwalkoperator} which obeys the following differential equation (up to a normalization factor):
\begin{equation}
  \label{eqn:propagation}
  \di_x M(x) = -\hat{\e}(x) M(x) .
\end{equation}
We can write $M(x)$ in the diagonal basis of $\hat{\e}(x)$ by introducing the operator
\begin{equation}
  \label{eqn:N}
  N(x) = \lp \bigoplus_{k=1}^{S(\mbV)} U^{\dag}_l(x) \rp M(x) \lp \bigoplus_{l=1}^{S(\mbV)} U_l(x) \rp .
\end{equation}
Eq.~\eqref{eqn:propagation} can then be rewritten as
\begin{equation}
  \label{eqn:Npropagation}
  \di_x N(x) = - \bigoplus_{l=1}^{S(\mbV)} D_l(x) N(x) - i \bigoplus_{l=1}^{S(\mbV)} \ls \di_x G_l(x), N_l(x) \rs .
\end{equation}
Note that since $M(0) = I$ then $N(0)=I$ and any solution $N(x)$ must be diagonal. Thus the total walk operator and the $\hat{\e}(x)$ operator are diagonal in the same basis.
\end{proof}

Lemmas~\ref{lem:closure},~\ref{lem:wedderburn}, and~\ref{lem:diagonal} combined give the full characterization of $M_1$ and $M_2$ operators achievable by our scheme:

\begin{thm}[Main result] \label{thm:main}
A continuous measurement using qubit probes and closed-loop feedback on the interaction Hamiltonian (as in Fig.~\ref{fig:circuit}) can realize any measurement $\lb M_1, M_2 \rb$ of the form
\begin{equation}
    M_1 = \bigoplus_{l=1}^{S(\mbV)} U^{\dag}_l \lp \bigoplus_{i=1}^{\mathrm{rank}\lp \mbB_l \rp} \lambda^{(l)}_i \Pi^{(l)}_i \rp U_l,
\end{equation}
where $M_2 = (I - M_1^{\dag}M_1)^{1/2}$ is diagonal in the same basis. The parameters $\lambda^{(l)}_i$ are real and contained in $(0,1)$ and $\Pi^{(l)}_i$ is a projector onto $1$, $2$, or $4$ basis states.
\end{thm}
\begin{proof}
Recall that the number of distinct diagonal entries possible in $D_l(x)$ is $\mathrm{rank} \lp \mbB_l \rp$. However, each distinct entry can appear $1$, $2$, or $4$ times depending on the particular representation from Table~\ref{tab:repr}. Using lemma~\ref{lem:diagonal} we can plug our solution for $D_l(x)$ into Eq.~\eqref{eqn:Npropagation} to find that the diagonal entries of $N(x)$ are $\exp ( \int_0^x \tanh (y - c^{(l)}_i) dy )$. The total walk operator $M(x)$ must then be
\begin{equation}
    M(x) = U^{\dag}_l(x) \lp \bigoplus_{i=1}^{\mathrm{rank} \lp \mbB_l \rp} e^{\int_0^x \tanh \lp y - c^{(l)}_i \rp dy} \Pi^{(l)}_i \rp U_l(x) .
\end{equation}
The endpoint operators $M_1$ and $M_2$ are proportional to $M(X)$ and $M(-X)$. Their diagonal entries are $\lambda^{(l)}_i$, which after renormalization approach $0$ when $c^{(l)}_i \rightarrow \infty$ and $1$ when $c^{(l)}_i \rightarrow -\infty$. 
\end{proof}

Note that in theorem~\ref{thm:main} the eigenvalues of $M_1$ and $M_2$ are restricted to lie in the open set $(0,1)$, not the closed set $[0,1]$. This is a consequence of the reversibility condition at the points $x=X-\d$ and $x=X+\d$. At these points, setting any eigenvalue of the total walk operator to $0$ would be effectively a projection, which is an irreversible operation for the random walk. However we can approach arbitrarily close to any such projective measurement.

To allow for direct comparisons with the scheme of~\cite{constham}, we provide the following corollary.

\begin{cor}[Spectrum of the measurement]
Given the ability to perform any unitary transformations directly before and after the continuous process of theorem~\ref{thm:main}, one can continuously decompose any measurement with $\sum_{l=1}^{S(\mbV)} \mathrm{rank} \lp \mbB_l \rp$ distinct singular values.
\end{cor}
\begin{proof}
The endpoint measurement operators $M_1$, $M_2$ in theorem~\ref{thm:main} can have up to  $\sum_{l=1}^{S(\mbV)} \mathrm{rank} \lp \mbB_l \rp$ distinct eigenvalues. We can decompose any pair of general endpoint operators $M_1$, $M_2$ using their polar decompositions $M_i = W_i (M_i^{\dag}M_i)^{1/2}$. Then, we can use a procedure like that of Figure~\ref{fig:circuit} to measure the positive Hermitian operators $(M_i^{\dag}M_i)^{1/2}$ and subsequently apply $W_i$ depending on the measurement result.
\end{proof}

In this work we've characterized the full class of continuous measurements achievable using a stream of probe qubits and a tunable interaction Hamiltonian. Given a set of linearly controlled Hamiltonian terms we provide a method to exhaustively list all continuous decompositions achievable with the control set. The class we find has a simple block-diagonal form, but results from a non-trivial application of the reversibility condition. Notably, measurements in this class have a quantifiably broader spectrum than in the case of a fixed interaction Hamiltonian.

Our work makes critical use of finite-dimensional Jordan algebras. This is surprising since these algebras have had little application elsewhere in quantum mechanics. Our model for continuous measurements does not include internal dynamics $H_S$ for the system or the probe, nor does it account for environment noise. In the presence of $H_S$, successive realizations of the continuous decomposition would yield inconsistent results unless $H_S$ commutes with the measurement operators.

The model presented here is still not the most general description of all continuous measurements realizable with a stream of probes. A completely general description would have to consider higher-dimensional probes, multiple outcomes to the weak measurement steps (as well as the endpoint measurements), and a more general reversibility condition. This is the subject of ongoing work. If Jordan algebras reappear in that scenario, then they will have found renewed application in quantum mechanics.

\begin{acknowledgments}
JF and TAB thank Daniel Lidar and Ognyan Oreshkov for useful discussions. This research was supported in part by the ARO MURI under Grant No. W911NF-11-1-0268.
\end{acknowledgments}

\bibliographystyle{apsrev4-1}
\bibliography{QuantumControl}

\end{document}